\documentclass[12pt]{article}
\usepackage[right=1in,left=1in,top=1in,bottom=1in]{geometry}
\usepackage{hyperref}
\hypersetup{colorlinks, citecolor=blue, filecolor=blue, linkcolor=blue, urlcolor=blue}
\usepackage{graphicx}
\usepackage{url}
\usepackage[round]{natbib}
\usepackage{amsmath,amsthm} 
\usepackage{engord}
\usepackage{float}
\usepackage{subfig}
\usepackage{tikz}
\usetikzlibrary{arrows.meta, positioning}
\usepackage{pdflscape}
\usepackage{booktabs}
\usepackage{pgfplots, verbatim}
\pgfplotsset{compat=1.14}
\pgfplotsset{every axis label/.append style={font=\tiny}}
\usepackage[labelsep=period]{caption} 

\usepackage{amssymb} 
\usepackage{multirow} 

\usepackage{xr}
\usepackage{bbm}

\usepackage{setspace}
\usepackage{sectsty}
\sectionfont{\large}
\subsectionfont{\normalsize}
\subsubsectionfont{\normalsize}
\newcommand{\R}{{\mathbb R}}

\newtheorem{theorem}{Theorem}
\newtheorem{prop}{Proposition}
\newtheorem{lemma}{Lemma}

\newtheorem{defn}{Definition}

\allowdisplaybreaks

\title{ \vspace*{-2.5cm} \hspace*{-0.5cm}Information About Other Agents in Mechanism Design}

\author{Eric Yan\thanks{Massachusetts Institute of Technology, eyan@mit.edu.} \footnote{This article subsumes my undergraduate thesis at Harvard University. I am indebted to my thesis advisor Eric Maskin for constant and invaluable guidance and Stephen Morris for generous discussions, as well as Michal Kowalik, Shengwu Li, Daniel Luo, Ellen Muir, Matthew Rabin, Madison Shirazi, and Charlie Yang for helpful comments and discussion. All errors are my own.}} 
\date{ June 2025\\}

\begin{document}

\bgroup
\let\footnoterule\relax
\begin{singlespace}
\maketitle

\begin{abstract}
    \noindent We study mechanism design settings where the planner has an interest in agents receiving noisy signals about the types of other agents. We show that additional information about other agents can eliminate undesired equilibria, making it helpful to a planner interested in full implementation — designing a mechanism for which \emph{every} equilibrium outcome is desirable. We provide a sufficient condition under which a social choice function that is not fully implementable when agents have no information about types of other agents can become fully implementable when agents have additional information.
\end{abstract}
\end{singlespace}

\egroup
\setcounter{page}{1}



\section{Introduction\label{sec:introduction}}
\noindent A key challenge of mechanism design is to design a game under which privately informed agents reach a desired outcome in equilibrium. The effectiveness of these mechanisms often depends critically on the information available to agents. Many recent papers have studied how varying what agents know about variables that directly affect their own payoffs — for example, the information a bidder knows about her own valuation for an object (\cite{BM22}) — affects a designer's goals. But in many settings agents naturally possess or can acquire information about other agents as well. Can a social planner strictly benefit from agents having additional information about other agents in the mechanism? We address this question within the canonical framework of Bayesian implementation of \cite{jackson91} (see also \cite{PS86} and \cite{PS89}). Our main finding is that a planner interested in \emph{full implementation} — designing a mechanism such that every equilibrium yields a desirable outcome — can strictly prefer some agents to receive noisy signals about other agents, and we give a sufficient condition under which this can occur.

As an example, consider a firm deciding whether to invest more in a product, maintain the status quo, or exit the market. The board (planner) relies on two analysts who privately observe positively correlated signals about whether demand is high or low. One analyst is more senior than the other, and the board is cautious: they want the firm to exit whenever the senior analyst observes a low signal, invest only when both analysts observe a high signal, and maintain the status quo if the senior analyst believes the demand is high and the junior analyst disagrees. Assuming the analysts care only about monetary payments, a natural way to ensure both analysts' truthful reporting of their signal as an equilibrium is to pay the analysts if and only if they agree. Unfortunately, this setup admits undesirable equilibria, such as analysts coordinating on a high or low demand report, regardless of their private signals. A key theoretical insight from the Bayesian implementation literature is that allowing agents to suggest alternative outcome rules can eliminate undesired equilibria if the agent proposing the alternate rule strictly benefits from the proposal at the undesired equilibrium but is weakly worse off if agents are at the preferred equilibrium. To illustrate our main point, we introduce two more agents — a CEO and a compliance officer — who do not have private information but have preferences over the firm's action. We exhibit a common prior and natural payoff structure under which the compliance officer can propose alternate rules that eliminate the equilibria where the analysts always coordinate on a message, but no agent is able to rule out an equilibrium where the analysts always report the demand is high when they observe a low signal and vice versa.

The CEO wants the firm to capture as large a share of the product’s market as possible, so she prefers investing to the status quo, and exiting is her least-preferred outcome. Suppose we give the CEO a binary signal about the analysts’ private information that achieves the following. In one realization, she assigns higher probability to both analysts seeing a high signal than to both seeing a low signal. In the other realization, she assigns higher probability to the senior analyst observing a high signal and the junior analyst observing a low signal than vice versa. In either realization, she can credibly propose an alternative outcome rule that swaps the firm actions prescribed to the respective pair of analyst signals. Each swap weakens her expected payoff when the analysts are truthful, so she never invokes it in the desired equilibrium. But it strictly improves her payoff when the analysts are engaged in the deceptive equilibrium where they always lie, giving her an incentive to propose the alternate rule and thereby break that equilibrium. 

Motivated by this example, this paper seeks conditions under which a planner would strictly prefer that agents have full-support signals about other agents’ \emph{types}. While the above example provides intuition for why such information can help, our starting result in Section~\ref{sec:partial} is that additional information never strictly benefits a planner interested only in \emph{partial implementation} — merely ensuring truth-telling is an equilibrium. The remainder of the paper examines how additional information about the types of other agents can help an agent signal untruthful equilibria to the planner. In section \ref{sec:Main example}, we formally describe the example from above. We show that the board's desired \emph{social choice function} (SCF) is not fully implementable when agents have no information about the types of other agents but is fully implementable when the CEO is given a noisy signal. The condition that rules out undesired equilibria in Bayesian implementation settings is \emph{Bayesian monotonicity}, which guarantees agents sufficient incentives to signal undesirable equilibria to the planner. Our theoretical analysis focuses heavily on how Bayesian monotonicity varies with the information agents have about other agents. Section \ref{sec:thms} presents a sufficient condition for an SCF that is not fully implementable when agents have no information about the types of other agents to be fully implementable when agents have more information.

An obvious concern with giving an agent additional information about types of other agents is that it can affect their own incentive to report truthfully to the mechanism. Indeed, the conclusion of section \ref{sec:partial} is that incentive compatibility becomes (weakly) harder to satisfy if agents receive additional information about other agents. To avoid eliminating truthful reporting as an equilibrium, we only modify the information available to agents for which truth-telling is a dominant strategy in the \emph{direct mechanism} of a given SCF. Our results do not speak to SCFs for which no agent has a dominant strategy in the direct mechanism.

To establish our main result, we note that a partially implementable SCF is not fully implementable if and only if there is an untruthful equilibrium of the direct mechanism such that no agent has an incentive to report the deception to the planner. The task is then to construct a signal that gives some agent $i$ the ability to eliminate these equilibria. In general, if truth-telling is a dominant strategy for player $i$, an untruthful equilibrium $\alpha$ of the direct mechanism will involve some subset of other agents playing deceptions. The signal we construct partitions the set of other agents' types into two groups. Each group in the partition must contain a type profile of other agents such that if the agent $i$ knew this type profile realized, she would prefer all other agents to report truthfully rather than deceive according to $\alpha$, regardless of her own type. If agent $i$ is confident enough that the realized type profile is in one group or the other, she can propose an alternate SCF that swaps the outcomes implemented under the truthful strategy versus the deception. By construction, this alternate SCF can only give agent $i$ a weakly lower payoff relative to the original if all other agents report truthfully, so agent $i$ would only propose it if $\alpha$ is being played. Whenever the signal is sufficiently accurate, agent $i$ can eliminate $\alpha$ as an equilibrium. We conclude by showing that this procedure can iteratively eliminate all untruthful equilibria of the direct mechanism.

\paragraph{Related Literature} This paper directly contributes to the literature on full Bayesian implementation — which began with \cite{PS86}, \cite{PS89}, and \cite{jackson91} — by shedding light on the relationship between Bayesian monotonicity and the information available to agents. Since agents are receiving informative signals about the state of the world, we also add to the extensive literature on the effects of changing information structures in mechanism design settings pioneered by \cite{BV02}. It  relates especially to the literature on information disclosure in mechanism design. In auction settings, \cite{ganuza04} studies the optimal way for a seller to disclose a costly public signal about a good's characteristic in a second-price auction, which is a component of bidders' private valuations, while \cite{BP07} investigate the joint design problem in which the seller controls the information bidders have about their own valuations and the auction rules. Other papers that have investigated the effects of information disclosure in auction settings include \cite{ES07}, and \cite{BM22}. Similar questions have also been investigated in screening settings by \cite{LS17} and \cite{roesler_buyer-optimal_2017}. Additionally, while our paper asks whether a planner can be interested in agents' information about other agents, one can also ask about an agent's incentives to acquire such information, as done in \cite{GP23}, who show that most mechanisms incentivize agents to acquire information about other agents in environments where agents have uncertainty about their own types as well as types of other agents. Finally, our emphasis on a designer's preferences over agent information structures connects with the recent literature on information design spurred by \cite{rayo2010optimal} and \cite{KamenicaGentzkow11}. Consistent with the information design literature, our analysis shows that a designer may prefer information structures strictly between no information and complete information.

\section{Model \label{sec:model}}
There are $n\geq 3$ agents $I=\{1,\dots,n\}$, a finite set of \emph{outcomes} $A$, and a finite number of \emph{payoff types} (henceforth types) $\theta_i\in \Theta_i$ for each agent $i.$ The state of the world $\theta\equiv(\theta_1,\dots,\theta_n)$ is in the set $\Theta\equiv\Theta_1\times\cdots\times \Theta_n.$ For any agent $i$, the type profile of agents other than $i$ is denoted $\theta_{-i}\in\Theta_{-i}\equiv\prod_{j\neq i}\Theta_j.$ There is a full-support \emph{common prior} $\mu$ over $\Theta$, and each agent knows her own type $\theta_i$ but cannot fully determine $\theta_{-i}.$ Each agent $i\in I$ has a utility function that depends only on $\Theta_i$, so $u_i:A\times \Theta_i\to \R$. Thus, for every type $\theta_i\in\Theta_i$, agent $i$ has a set of preferences over the outcomes in $A.$ This is the case of \emph{private values}.

The \emph{planner} wishes to implement a deterministic \emph{social choice function} (SCF) $F:\Theta\to A.$ The planner commits to a mechanism $\Gamma=\langle M,h\rangle$, such that $M=\prod_{i=1}^nM_i$ consists of a message space $M_i$ for each agent $i$ and $h:M\to A$ is the outcome function. A strategy profile $s=(s_1,\dots,s_n)$ specifies for each agent $i$ a function $s_i:\Theta_i\to M_i$.\footnote{The restriction to pure strategies is only for expositional simplicity and is discussed in footnote \ref{footnote:mixed}.} At any state $\theta$, let $s(\theta)=(s_1(\theta_1),\dots,s(\theta_n))$ be the message profile sent by all agents, and $s_{-i}(\theta_{-i})=(s_j(\theta_j))_{j\neq i}$ denote the messages sent by agents other than $i.$ Let $B^\Gamma$ be the set of Bayesian Nash equilibria for a mechanism $\Gamma.$ If there is an $s\in B^\Gamma$ such that $h\circ s=F$ then $\Gamma$ \emph{partially implements} $F$, and $F$ is considered \emph{partially implementable}. If $\Gamma$ partially implements $F$, and $B^\Gamma$ contains only one equilibrium outcome, then $\Gamma$ \emph{fully implements} $F$ and $F$ is considered \emph{fully implementable}.\footnote{Standard implementation theory allows the planner to specify a \emph{social choice set} (SCS). Full implementation of an SCS is equivalent to full implementability of each SCF in the set under a full-support prior.} We assume the environment is \emph{economic}: For any state $\theta\in\Theta$ and any SCF $F$, at least two agents would prefer to alter the outcome prescribed by $F.$

Classical mechanism design results pin down two key conditions for the implementability of an SCF $F$.

\begin{defn}
    An SCF $F$ is (Bayesian) \textbf{incentive-compatible} if and only if for all agents $i$ and all $\theta_i,\theta_i'\in\Theta_i,$ \[\mathbb{E}(u_i(F(\theta_i,\theta_{-i}),\theta_i)\mid \theta_i) \geq \mathbb{E}(u_i(F(\theta_i',\theta_{-i}),\theta_i)\mid \theta_i),\] where the expectation is taken over the belief of type $\theta_i$ about $\Theta_{-i},$ which is formed from the common prior $\mu$. 
\end{defn}

From the revelation principle we know that an SCF $F$ is partially implementable iff $F$ is incentive-compatible. For an incentive-compatible $F$ to be fully implementable, it must also satisfy \emph{Bayesian monotonicity}, as studied by \cite{PS86, PS89, jackson91}.

\begin{defn}
    A \textbf{deception for agent} $i$ is a function $\alpha_i:\Theta_i\to \Theta_i.$ A \emph{deception} $\alpha=(\alpha_1,\dots,\alpha_n)$ consists of a deception $\alpha_i$ for each agent $i\in I.$ 
\end{defn}

\begin{defn}
    Let $\alpha$ be a deception and $$\alpha(\theta)=(\alpha_1(\theta_1),\dots,\alpha_n(\theta_n))$$ and $$\alpha_{-i}(\theta_{-i}) = (\alpha_1(\theta_1),\dots,\alpha_{i-1}(\theta_{i-1}),\alpha_{i+1}(\theta_{i+1}), 
    \dots,\alpha_n(\theta_n)).$$ An SCF $F$ satisfies \textbf{Bayesian monotonicity} if for all deceptions $\alpha$ such that $F\circ\alpha\neq F$ there is an agent $i$ and a function $y:\Theta_{-i}\to A$ such that \[\mathbb E(u_i(F(\theta_{-i},\theta_i),\theta_i)\mid \theta_i)\geq \mathbb E(u_i(y(\theta_{-i}),\theta_i)\mid\theta_i)\] for all $\theta_i\in \Theta_i$ and \[\mathbb E(u_i(F(\alpha(\theta_{-i},\theta_i')),\theta_i')\mid\theta_i')< \mathbb E(u_i(y(\alpha_{-i}(\theta_{-i})),\theta_i')\mid\theta_i')\] for some $\theta_i'\in\Theta_i.$
\end{defn}

It was shown in \cite{jackson91} that incentive compatibility and Bayesian monotonicity are necessary and sufficient for an SCF $F$ to be fully implementable in economic environments with at least 3 agents.\footnote{\cite{jackson91} identifies a closure condition necessary for the implementation of SCSs, but the condition is vacuous for the case of a single-valued SCF.}\footnote{\cite{jackson91} characterizes full implementation in pure strategies, i.e. every pure-strategy equilibrium of the mechanism is desirable. If agents can use mixed strategies, full implementation in our setting should demand that every realization of the strategy profile yields an outcome that aligns with $F.$ \cite{SV2010} identify a mixed Bayesian monotonicity condition that together with incentive compatibility is necessary and sufficient for full implementation of an SCF in economic environments when agents can use mixed strategies. The condition is equivalent to Bayesian monotonicity in the case of a single-valued SCF, implying that all our results hold if we allow agents to use mixed strategies.\label{footnote:mixed}} While incentive compatibility guarantees that truth-telling is \emph{an} equilibrium of the direct mechanism, Bayesian monotonicity is the condition that allows the planner to design a mechanism that can eliminate undesired equilibria. The intuition behind the condition, which we will reference throughout this paper, is the following. Suppose agents are initially at an equilibrium that uses a deception $\alpha$ such that $F\circ\alpha\neq F.$ The planner can enlarge the message space of each agent beyond her own type space to allow any individual agent $i$ to specify an alternate SCF. (The function $y$ in the definition of Bayesian monotonicity can be extended to $\Theta$ by implementing the same outcome regardless of agent $i$'s type.) The planner can commit to implementing the alternate SCF only if it appears to weakly harm agent $i$, regardless of her type, if the agents were playing the truthful equilibrium. Bayesian monotonicity gives at least one type of agent $i$ an incentive to specify such an alternate SCF because she would prefer the alternate SCF to $F$ under the deception $\alpha$, ruling out the deception as an equilibrium of a larger mechanism.

The beliefs agents have regarding the types of other agents is an important component of both incentive compatibility and Bayesian monotonicity. We study how these conditions change when agents have additional information about the types of other agents. 

Suppose each agent $i$ can access a signal \[\sigma_i:\Theta_{-i}\to \Delta(\zeta_{i}),\] where $\zeta_i$ is a finite set of signal realizations and $\sigma_i(\theta_{-i})$ is the signal observed by agent $i$ when the type profile of other agents is $\theta_{-i}.$ Assume that $\sigma_i$ has full support over $\zeta_{i}$ for every agent $i.$ After observing the realization of $\sigma_i$ and her type $\theta_i$, agent $i$ updates her beliefs regarding $\Theta_{-i}$ (and $\Theta$) according to Bayes' rule. In particular, suppose $\zeta_i=\{A_{i1},\dots,A_{ik}\}$ and agent $i$ has type $\theta_i.$ Upon observing $\sigma_i=A_{ij}$, agent $i$ places probability \[\mu(\theta_{-i}\mid\theta_i,\sigma_i=A_{ij})=\frac{\mu(\theta_{-i}\mid\theta_i)\sigma_i(\theta_{-i})(A_{ij})}{\sum_{\theta_{-i}'}{\mu(\theta_{-i}'\mid\theta_i)\sigma_i(\theta_{-i}')(A_{ij})}}\] on type $\theta_{-i}.$ 

We follow the standard timing used in implementation theory with the requirement that the signal realizations are observed sometime prior to when agents choose their strategies in the mechanism so that the informational content of the signal is factored into their strategies.\footnote{Technically one can interpret an agent's signal $\sigma_i$ about other agents as part of agent $i$'s private information, so if $\sigma_i$ has possible realizations in $\zeta_i=\{A_{i1},\dots,A_{ik}\}$ an implementing mechanism should ask each agent $i$ to specify a strategy for each $(\theta_i,A_{ij})$ pair. Our setup is implicitly analyzing an equivalent formulation of this mechanism in which the uncertainty coming from the signals $\{\sigma_i\}$ is resolved first and asking whether $F$ is partially or fully implementable in each resulting subgame of the full mechanism.} When the signals $\{\sigma_i\}_{i\in I}$ provide no information, the model is equivalent to the classical Bayesian implementation problem studied by \cite{jackson91}. The main question we address is whether and when an SCF $F$ that is not partially (resp. fully) implementable when $\sigma_i$ is noninformative for all $i$ can become partially (resp. fully) if one or more of the $\sigma_i$ is nontrivial. We are agnostic on where the signals $\{\sigma_i\}_{i\in I}$ come from. The nature of the question is one of comparative statics on the implementability of an SCF $F$ with respect to an increase in information one more agents $i$ has regarding $\Theta_{-i}$.\footnote{Since agents initially do not have any information about $\Theta_{-i}$ aside from the common prior $\mu$, any nontrivial signal $\sigma_i$ represents an increase in agent $i$'s information with respect to the Blackwell order.} In the remainder of this paper we will say that $F$ is  \emph{partially implementable with respect to} $\{\sigma_i\}_{i\in I}$ if truth-telling is an equilibrium in the direct mechanism in each realization of the signals $\{\sigma_i\}_{i\in I}$. Similarly define \emph{fully implementable with respect to} $\{\sigma_i\}_{i\in I}.$ 

Our setup makes several substantive assumptions. Other than requiring all $\sigma_i$ to have full support, we assume there are no restrictions over the structure of the signal agent $i$ can receive. Thus, agents do not have private information that can constrain the types of signals they can process. Additionally, we assume the signals are not costly in order to focus on the role of the additional information regarding other agents' types. Finally, the structure of the signals $\{\sigma_i\}_{i\in I}$ is assumed to be common knowledge. The restriction that each $\sigma_i$ has full support is avoids complications to the type space $\Theta$ after agents receive their signals. In particular, any time an agent receives a signal that allows her to completely rule out any type profile, her set of possible types strictly enlarges. For example, if agents received a perfect signal about $\theta_{-i}$, each agent's private type is the full state of the world $\theta$ in this model, which reduces the problem to the classical Nash implementation problem solved by \cite{Maskin99}.\footnote{In particular, incentive compatibility is a trivial constraint in the complete information setting.}\footnote{See \cite{bergemann_robust_2005} for a more detailed discussion of how changes to the type space affect implementation of SCFs.}
\section{Partial implementation and the direct mechanism}
\label{sec:partial}
In this section we show that providing additional information to agents about the types of other agents is not helpful for the planner interested only in partial implementation and set up our full implementation analysis. We want to show that if $F$ is not implementable with respect to an uninformative signal structure, then it is not implementable with respect to any signal structure $\{\sigma_i\}_{i\in I}$. For the remainder of this paper, write $\phi_i$ to mean an uninformative signal about $\Theta_{-i}$ for each $i.$

Given any SCF $F$ we can consider the \emph{direct mechanism} defined by $\Gamma_F=\langle M,h\rangle,$ where $M_i=\Theta_i$ for each $i$ and $h=F.$ Each agent $i$ has a \emph{truth-telling} strategy $s:\Theta_i\to\Theta_i$ given by $s(\theta_i)=\theta_i.$ The revelation principle says that $F$ is partially implementable with respect to $\{\phi_i\}$ iff for each agent $i\in I$ and $\theta_i,\theta_i'\in \Theta_i$,

\[\sum_{\theta_{-i}\in \Theta_{-i}} \mu(\theta_{-i}\mid\theta_i)u_i(F(\theta_i, \theta_{-i}),\theta_i)\geq \sum_{\theta_{-i}\in \Theta_{-i}} \mu(\theta_{-i}\mid\theta_i)u_i(F(\theta_i', \theta_{-i}),\theta_i).\tag{i}\label{eq:noinfoIC}\]

If, however, $\sigma_i\in \{A_{i1}, \dots, A_{ik}\}$ contained  information about $\Theta_{-i}$, partial implementation requires $F$ to be incentive-compatible with respect to the beliefs of agent $i$ in every realization of $\sigma_i.$ In other words, for every $A_{i\ell}$ and $\theta_i,\theta_i'\in\Theta_i$ it must be that \[\sum_{\theta_{-i}\in \Theta_{-i}}\mu(\theta_{-i}\mid\theta_i, \sigma_i=A_{i\ell})u_i(F(\theta_i,\theta_{-i}),\theta_i)\geq \sum_{\theta_{-i}\in \Theta_{-i}}\mu(\theta_{-i}\mid\theta_i, \sigma_i=A_{i\ell})u_i(F(\theta_i',\theta_{-i}),\theta_i).\tag{ii}\label{eq:infoIC}\]

The law of total expectations implies that~\eqref{eq:noinfoIC} is just a weighted average of the equations ~\eqref{eq:infoIC} for $1\leq \ell \leq k.$ This implies the following proposition.
\begin{prop}
    Suppose $F$ is an SCF that is partially implementable with respect to some signal structure $\{\sigma_i\}_{i\in I}$. Then $F$ is partially implementable with respect to $\{\phi_i\}_{i\in I}$.
\end{prop}

In other words, it can only become harder for $F$ to satisfy incentive compatibility if an agent receives an informative signal about other agents' types, compared to the no information case. We wish to avoid ruling out the equilibrium specified by $F$ while making certain agents' information structures more informative. Thus for the remainder of this paper we focus on the information of agents for which truth-telling is a dominant strategy in the direct mechanism.

It is well-known that the direct mechanism generally has undesired equilibria. The mechanism exhibited by \cite{jackson91} to fully implement an SCF $F$ allows each agent to report, among other information, an alternate SCF that can be used to signal a deception to the planner, provided $F$ is Bayesian monotonic. We introduce the following notion to keep track of which agents have an incentive to alert the planner of a given deception.

\begin{defn}
    Let $F$ be an SCF and $\alpha$ be a deception such that $F\circ \alpha \neq F$. Say that $\alpha$ is \textbf{undermined} by agent $i$ (or agent $i$ undermines $\alpha$) if there is a function $y:\Theta_{-i}\to A$ such that \[\mathbb E(u_i(F(\theta_{-i},\theta_i),\theta_i)\mid \theta_i)\geq \mathbb E(u_i(y(\theta_{-i}),\theta_i)\mid\theta_i)\] for all $\theta_i\in \Theta_i$ and \[\mathbb E(u_i(F(\alpha(\theta_{-i},\theta_i')),\theta_i')\mid\theta_i')< \mathbb E(u_i(y(\alpha_{-i}(\theta_{-i})),\theta_i')\mid\theta_i')\] for some $\theta_i'\in\Theta_i.$
\end{defn}

The inequalities are exactly the ones from the definition of Bayesian monotonicity. Since a deception is equivalent to a pure-strategy profile in the direct mechanism, the only deceptions $\alpha$ such that $F\circ\alpha\neq F$ that need to be undermined in order to guarantee the full implementation of $F$ are (untruthful) equilibria of the direct mechanism. For any incentive-compatible SCF $F$, write $\mathcal{A}_F$ to be the set of deceptions $\alpha$ such that $F\circ\alpha\neq F$ and constitute a pure-strategy equilibrium of the direct mechanism when agents do not have additional information about other agents. Let $\mathcal{A}_F^u\subset \mathcal{A}_F$ be the set of deceptions that can be undermined by some agent $i,$ again when agents have no additional information about other agents. Then we have the following.

\begin{lemma}
\label{lem:revelation}
Let $F$ be an incentive-compatible SCF. Then $F$ is fully implementable if and only if $\mathcal{A}_F^u= \mathcal{A}_F$.
\end{lemma}
The lemma follows from the mechanism constructed in the full Bayesian implementation result of \cite{jackson91}. It allows us to focus our attention on a small and well-defined set of deceptions rather than the unruly space of all possible deceptions suggested by the definition of Bayesian monotonicity.

\section{An illustrative example \label{sec:Main example}}

The goal of this paper is to provide a sufficient condition for a non-fully-implementable SCF to be fully implementable when agents can receive noisy information about other agents' types. In this section, we illustrate an example that shows how an agent for which truthful reporting is a dominant strategy can use information about other agents to signal deceptions to the planner.

Suppose a firm selling a product is deciding whether to invest more in the product ($I$), maintain the status quo ($S$), or exit the market ($E$). To decide what to do, the board (planner) designs a mechanism between a senior analyst (agent 1), junior analyst (agent 2), the CEO (agent 3), and a compliance officer (agent 4). Each analyst $i$ privately observes a signal in $\{H_i,L_i\}$ regarding whether demand for the product is high ($H_i$) or low ($L_i$), which constitute's agent $i$'s private type $\theta_i$. The CEO and compliance officer have only one type (or equivalently do not have private information), and the analysts' types are correlated according to the common prior distribution given by \[\mu(\theta_1=H_1,\theta_2=H_2)=\mu(\theta_1=L_1,\theta_2=L_2)=\frac13,\] and   \[\mu(\theta_1=H_1,\theta_2=L_2)=\mu(\theta_1=L_1,\theta_2=H_2)=\frac16.\] An SCF maps the pair of analysts’ type to a triple $(y,t_1,t_2)$ consisting of a firm action $b\in \{I, S, E\}$ and non-negative transfers $t_1,t_2$ received by the analysts. The firm has 2 indivisible units of money that can be paid out to analysts, so the outcome set is \[A= \{I,S,E\}\times \{(t_1,t_2)\in \mathbb{N}^2\mid t_1+t_2\leq 2\}.\] All agents have type-independent preferences, so we omit types from the utility functions in this example. Each analyst's utility is equal to his transfer. The CEO's utility function is given by \[u_3(b,t_1,t_2)=9\mathbbm{1}(b=I)+6\mathbbm{1}(b=S)+3\mathbbm{1}(b=E)-t_1-t_2,\] reflecting that she always prefers investing to keeping status quo to exiting, and within each action prefers to give spend less money on the analysts. Finally, the compliance officer's utility is \[u_4(b,t_1,t_2)=2\mathbbm{1}(b=S)-\mathbbm{1}(b=I),\] reflecting a preference for keeping status quo to exiting to investing. The board’s objective is given by the SCF
\[
F(\theta_1,\theta_2)=
\begin{cases}
(I,1,1) &\text{if }(\theta_1,\theta_2)=(H_1,H_2),\\[4pt]
(S,0,0) &\text{if }(\theta_1,\theta_2)=(H_1,L_2),\\[4pt]
(E,0,0) &\text{if }(\theta_1,\theta_2)=(L_1,H_2),\\[4pt]
(E,1,1) &\text{if }(\theta_1,\theta_2)=(L_1,L_2).
\end{cases}
\]
In other words, the board wants to pay the analysts \((1,1)\) only when they deliver the same recommendation. The firm values the senior analyst's opinion more and is risk-averse, exiting whenever the senior analyst reports low demand and maintains the status quo if the senior analyst reports high demand but the junior analyst disagrees. The board only wants to invest if both analysts agree the demand is high. In the remainder of this section, we show that $F$ is not fully implementable under an uninformative signal structure and construct a signal that can be provided to the CEO about the analysts' that allows the board to fully implement $F.$

First consider the implementability of $F$ without additional signals. As noted in Lemma \ref{lem:revelation}, $F$ is fully implementable if and only if every equilibrium $\alpha$ of the direct mechanism such that $F\circ\alpha\neq F$ can be undermined by some agent. The direct mechanism, which allows only the analysts to have nontrivial strategy spaces in this case, has exactly four pure-strategy Bayes Nash equilibria. Note that neither the CEO nor the compliance officer can affect the outcome in the direct mechanism, since they have only one type and truth-telling is trivially their dominant strategy.

\begin{enumerate}
    \item \emph{Truth-telling}: Analysts report truthfully, so $s_i(H_i)=H_i$ and $s_i(L_i)=L_i$ for $i=1,2.$
    \item \emph{Always-$H$}: $s_i(H_i)=s_i(L_i)=H_i$.
    \item \emph{Always-$L$}: $s_i(H_i)=s_i(L_i)=L_i$.
    \item \emph{Always-lie}: $s_i(H_i)=L_i$ and $ s_i(L_i)=H_i$.
\end{enumerate}

\begin{prop}
The compliance officer can undermine both the always-$H$ and always-$L$ equilibria, but no agent can undermine the always-lie equilibrium.
\end{prop}
\begin{proof}
Recall that an agent can undermine a deception (undesired equilibrium of the direct mechanism) if she can propose an alternate SCF that weakly harms her payoff, regardless of her own type, under truth-telling but strictly benefits her payoff under the deception. If the analysts are playing the always-$H$ deception, then the compliance officer can propose the following alternate SCF: \[
y(\theta_1,\theta_2)=
\begin{cases}
(S,0,0) &\text{if }(\theta_1,\theta_2)=(H_1,H_2),\\[4pt]
(I,1,1) &\text{otherwise}.
\end{cases}
\]
It is readily verified that $y$ gives the compliance officer the same expected utility as $F$ if the analysts were truthfully reporting but strictly benefits her if the analysts are playing always-$H$, implying that the compliance officer can undermine the always-$H$ deception. A similar argument shows the compliance officer can undermine the always-$L$ deception. To show $F$ is not fully implementable, it remains to show no agent can undermine the always-lie deception. To see this, first consider the senior analyst's perspective. Any alternate SCF $y$ he can propose must only depend on $\theta_2$, so let $y(H_2)=a_1$ and $y(L_2)=a_2.$ If the senior analyst can undermine always-lie, $y$ must weakly harm his utility under truth-telling, so $a_1,a_2$ must be chosen such that \[\frac23 u_1(a_1)+\frac13u_1(a_2)\leq \frac23\] and \[\frac13 u_1(a_1)+\frac23u_1(a_2)\leq \frac23.\] But for the senior analyst to have an incentive to propose $y$ if the analysts are playing always-lie, $y$ must satisfy either \[\frac23 u_1(a_1)+\frac13u_1(a_2)> \frac23\] or \[\frac13 u_1(a_1)+\frac23u_1(a_2)> \frac23,\] contradicting the previous requirement. Thus the senior analyst cannot undermine the always-lie deception. A similar argument applies to the junior analyst. Finally, both the CEO and compliance officer are always indifferent between truth-telling and always-lie under any SCF, so neither agent can undermine always-lie.\footnote{This indifference is not necessary in this example or in general but simplifies the exposition.}
\end{proof}
Therefore, $F$ is not fully implementable under an uninformative signal structure.
\begin{prop}
$F$ is fully implementable with respect to $\{\sigma_i\}$, where $\sigma_1,\sigma_2,$ and $\sigma_4$ are uninformative and $\sigma_3$ has support in $\{0,1\}$ and satisfies $$Pr(\sigma_3=0\mid \theta_{-3}=(H_1,H_2))=Pr(\sigma_3=0\mid \theta_{-3}=(L_1,H_2))=\tau$$ and $$Pr(\sigma_3=1\mid \theta_{-3}=(H_1,L_2))=Pr(\sigma_3=1\mid \theta_{-3}=(L_1,L_2))=\tau$$ for any $\tau>1/2$.
\end{prop}

The signal $\sigma_3$ essentially provides a noisy signal to the CEO regarding whether the junior analyst believes the demand is high or low. This has many possible interpretations, such as allowing the CEO to monitor the junior analyst's work more closely.
\begin{proof}
Suppose $\sigma_3=0$ is realized. Then the CEO believes it is more likely that the type profile of the analysts belongs to $\{(H_1,H_2),(L_1,H_2)\}$ rather than $\{(H_1,L_2),(L_1,L_2)\}$. Since $\tau>1/2,$ she believes it is more likely for the analysts' type profile to be $(H_1,H_2)$ than $(L_1,L_2).$ Suppose she proposes the alternate SCF  \[
y(\theta_1,\theta_2)=
\begin{cases}
(E,1,1) &\text{if }(\theta_1,\theta_2)=(H_1,H_2),\\[4pt]
(S,0,0) &\text{if }(\theta_1,\theta_2)=(H_1,L_2),\\[4pt]
(E,0,0) &\text{if }(\theta_1,\theta_2)=(L_1,H_2),\\[4pt]
(I,1,1) &\text{if }(\theta_1,\theta_2)=(L_1,L_2).
\end{cases}\] which swaps the outcomes corresponding to  $(H_1,H_2)$ and $(L_1,L_2).$ Compared to $F$, the SCF $y$ gives her a lower utility if analysts are truth-telling but improves her utility if analysts are always lying. Thus the CEO can undermine always-lie if $\sigma_3=0$ is realized. Similarly, if $\sigma_3=1$ is realized, the CEO believes $(H_1,L_2)$ is a more likely analyst profile than $(L_1,H_2)$. A similar alternate $y$ swapping the outcomes corresponding to the two profiles allows her to undermine always-lie. Therefore the CEO can undermine always-lie in either realization of $\sigma_3.$ Since $\sigma_3$ only changes the CEO's information, the deceptions always-$H$ and always-$L$ can still be undermined by the compliance officer in either realization. It follows that $F$ is fully implementable with respect to $\{\sigma_i\}.$
\end{proof}

This example illustrates how an agent for which truth-telling is a dominant strategy can undermine deceptions with additional information about the types of other agents. When an SCF $F$ fails to be fully implementable, the culprit is at least one equilibrium $\alpha$ of the direct mechanism that the planner dislikes and cannot be undermined. Suppose truth-telling is a dominant strategy for agent $i$, and she always prefers the outcome associated with the type profile $\theta_{-i}'$ to the type profiles in $\alpha_{-i}(\theta_{-i}')$ — regardless of her own type and what she reports. Then a simple ``outcome" swap can kill $\alpha$ whenever she believes $\theta_{-i}'$ is more likely to be the true type profile of other agents than the type profiles in $\alpha_{-i}^{-1}(\theta_{-i}')$. A suitably constructed signal can ensure she holds such a belief. Additionally, as long as agent $i$ can receive a signal that ensures such a type profile $\theta_{-i}'$ exists in each realization, then an undesirable equilibrium $\alpha$ can be killed.
\section{A sufficient condition}
\label{sec:thms}

We begin by formalizing the intuition provided in the previous section into the following condition.
\begin{defn}
\label{defn:potentially}
Let $F$ be an SCF and $\alpha\in\mathcal{A}_F$. Agent $i$ can \textbf{potentially undermine} $\alpha$ if truth-telling is a dominant strategy for agent $i$ and there exist distinct type profiles $\theta_{-i}',\theta_{-i}''$ such that $\alpha_{-i}(\alpha_{-i}(\theta_{-i}'))\neq \theta_{-i}''$, $\alpha_{-i}(\alpha_{-i}(\theta_{-i}''))\neq \theta_{-i}'$, $\alpha_{-i}(\theta_{-i}')\neq \alpha_{-i}(\theta_{-i}'')$ and the inequalities $$u_i(F(\theta_{-i}',\theta_i'),\theta_i')\geq u_i(F(\alpha_{-i}(\theta_{-i}'),\theta_i'),\theta_i'),$$ $$u_i(F(\theta_{-i}'',\theta_i'),\theta_i')\geq u_i(F(\alpha_{-i}(\theta_{-i}''),\theta_i'),\theta_i')$$ hold for all $\theta_i'\in\Theta_i.$ Furthermore, for each inequality there is a type $\theta_i$ that makes it strict. 
\end{defn}

In words, the definition requires two type profiles of other agents that agent $i$ prefers to be truthfully reported rather than manipulated by $\alpha$ in all states of the world, and for each type profile this preference is strict for at least one of agent $i$'s types. We note that the preferences do not need to be strict for the \emph{same} type $\theta_i$. The requirements that \[\alpha_{-i}(\alpha_{-i}(\theta_{-i}'))\neq \theta_{-i}'', \alpha_{-i}(\alpha_{-i}(\theta_{-i}''))\neq \theta_{-i}', \text{ and }\alpha_{-i}(\theta_{-i}')\neq \alpha_{-i}(\theta_{-i}'')\] limit the interactions $\theta_{-i}'$ and $\theta_{-i}''$ have with each other in the deception $\alpha$, which allows for a natural partition of $\Theta_{-i}$. As the name suggests, an agent $i$ that can potentially undermine a deception $\alpha$ can undermine $\alpha$ with additional information about the types of other agents. 


\begin{lemma}
\label{lem:main}
Let $F$ be an incentive-compatible SCF $\alpha$ be a deception that agent $i$ can potentially undermine. Then there is a binary signal $\sigma_i:\Theta_{-i}\to\Delta(\{0,1\})$ such that agent $i$ can undermine $\alpha$ after receiving either realization of $\sigma_i.$
\end{lemma}

In particular, for $\theta_{-i}',\theta_{-i}''$ as in Definition \ref{defn:potentially}, the signal informs agent $i$, with sufficiently high accuracy, whether or not the type profile of other agents is in the set \[\{\theta_{-i}',\alpha_{-i}(\theta_{-i}'')\}\cup \alpha_{-i}^{-1}(\theta_{-i}').\] We emphasize that the signal does not need to have a \emph{specific} level of accuracy: Any sufficiently accurate signal will accomplish the same objective. For an incentive-compatible $F$ that is not fully implementable, the lemma suggests a systematic way to eliminate all deceptions in $\mathcal{A}_F\setminus \mathcal{A}_F^u$ (which must be non-empty by lemma \ref{lem:revelation}) with additional information. In particular, we have the following result.

\begin{theorem}
\label{thm:main}
Let $F$ be an incentive-compatible SCF. If there is a subset $J\subset I$ such that each $\alpha\in \mathcal{A}_F\setminus \mathcal{A}_F^u$ can be potentially undermined by a distinct agent $j\in J$, and every $\alpha\in \mathcal{A}_F^u$ can be undermined by some agent $i\in I\setminus J$, then $F$ is fully implementable with respect to some signal structure $\{\sigma_i\}_{i\in I}.$ 
\end{theorem}

Establishing the result hinges on showing that the set of Bayes Nash equilibria of the direct mechanism remain unchanged after agents receive the signals constructed and applying lemmas \ref{lem:revelation} and \ref{lem:main}. In particular, the signals do not create additional undesirable equilibria of the direct mechanism that were not originally present. The theorem requires a separation of the agents who can undermine deceptions in $\mathcal A_F^u$ from the agents who can potentially undermine deceptions in $\mathcal A_F\setminus \mathcal A_F^u$ because it is not possible in general to guarantee that an agent who undermines a deception $\alpha\in \mathcal A_F^u$ without additional information can still undermine $\alpha$ after receiving an informative signal about other agents. This implicitly requires $|\mathcal A_F\setminus \mathcal A_F^u|\leq n-1$. 


\section{Conclusion \label{sec:conclusion}}

The main finding of the paper is that increasing the information agents have about the types of other agents can serve to eliminate undesired equilibria. We provide a sufficient condition under which a planner interested in full implementation will prefer agents have information about other agents beyond the common prior. The condition is straightforward to verify, as it requires only some strict preferences from some agents over type profiles of other agents and does not rely on any specific payoff structures. A future direction can be to relax the full-support signal assumption. Doing so would necessitate careful modeling of type spaces but could prove even more useful for the planner: Such signals can immediately enlarge the strategy spaces for agents in the direct mechanism, which can not only eliminate undesired equilibria but also cure violations of incentive compatibility.
\appendix

\section{Appendix: Omitted proofs}
\label{appendix:proofs}
\textbf{Proof of Proposition 1.}
\begin{proof}
Fix an agent $j.$ Since every agents' beliefs and every signal has full support, the only possible signal that affects agent $j$'s expected payoff is $\sigma_j.$ Suppose $\sigma_j\in\{A_{j1},\dots,A_{jk}\}$ and let $q_\ell$ be the unconditional probability that $\sigma_j=A_{j\ell}$ (so $q_\ell$ is averaged over all possible states). If $F$ is incentive-compatible with respect to agent $j$ after agent $j$ receives this signal, then \[q_\ell \mathbb E(u_j(F(\theta_{-j},\theta_j),\theta_j)\mid \theta_j, \sigma_j=A_{j\ell})\geq q_\ell \mathbb E(u_j(F(\theta_{-i},\theta_j'),\theta_j)\mid\theta_j, \sigma_j=A_{j\ell})\] for all $1\leq\ell\leq k.$ By the law of iterated expectations, \[\sum_{\ell=1}^kq_\ell \mathbb E(u_j(F(\theta_{-j},\theta_j),\theta_j)\mid \theta_j, \sigma_j=A_{j\ell})=\mathbb E(u_j(F(\theta_{-j},\theta_j),\theta_j)\mid \theta_j),\] and \[\sum_{\ell=1}^k q_\ell \mathbb E(u_j(F(\theta_{-i},\theta_j'),\theta_j)\mid\theta_j, \sigma_j=A_{j\ell})=\mathbb E(u_j(F(\theta_{-i},\theta_j'),\theta_j)\mid\theta_j).\] Therefore, it must be that \[\mathbb E(u_j(F(\theta_{-j},\theta_j),\theta_j)\mid \theta_j)\geq \mathbb E(u_j(F(\theta_{-i},\theta_j'),\theta_j)\mid\theta_j),\] so $F$ was incentive-compatible prior to agent $j$ receiving $\sigma_j$.
\end{proof}

\textbf{Proof of Lemma 1.}
\begin{proof}
  The only if direction is direct from the definitions. For the if direction, consider the following mechanism $\Gamma=\langle M,h\rangle$ to implement the SCF $F.$ The mechanism, slightly adapted from \cite{jackson91} to implement a single SCF $F$ rather than a set of SCFs, was shown to fully implement any SCF $F$ satisfying incentive compatibility and Bayesian monotonicity in economic environments. Let $S=\max_i|\Theta_i|$ and $K=n(S+1).$ Let $V=\{0,1,\dots,S^2\}^K,$ so $V$ consists of $K$-dimensional integer vectors with entries between $0$ and $S^2.$ Let $X$ be the space of all possible SCFs $F:\Theta\to A.$  Let $$M_i=\Theta_i\times\{\{\varnothing\}\cup V\} \times X\times \{\{\varnothing\}\cup X\}.$$  For any message profile $m=(m_1,\dots,m_n)\in M$, let $\theta^m\in \Theta$ be the type profile given by the first entry of every agent's message. The following rules specify $h$. 
  
  \underline{Rule 1:} Suppose there is an $i\in I$ such that every $m_j$ for $j\neq i$ is of the form $(\cdot, \varnothing, \cdot,\varnothing)$. Additionally, $m_i$ is of the form $(\cdot, \varnothing, \cdot,\varnothing)$, or $(\cdot, \cdot, \cdot,\varnothing).$ Then $h(m)=F(\theta^m).$

    \underline{Rule 2:} Suppose there is an $i\in I$ such that every $m_j$ for $j\neq i$ is of the form $(\cdot, \varnothing, \cdot,\varnothing)$. Additionally, $m_i$ is of the form $(\cdot,\cdot, \cdot, y)$ for some $y\in X$. Let $\theta_i$ be the first entry of $m_i,$ and define $y_{\theta_i}\in X$ by $y_{\theta_i}(\theta)=y(\theta_{-i},\theta_i)$ for all $\theta\in\Theta.$ Then
    \[h(m)=  \begin{cases}
        y(\theta^m) & \text{if } \mathbb E(u_i(F(\theta_{-i},\theta_i'),\theta_i')\mid\theta_i')\geq \mathbb E(u_i(y_{\theta_i}(\theta_{-i},\theta_i'),\theta_i')\mid\theta_i') \text{ for all } \theta_i'\in\Theta_i\\
        F(\theta^m) & \text{if } \mathbb E(u_i(y_{\theta_i}(\theta_{-i},\theta_i'),\theta_i')\mid\theta_i')> \mathbb E(u_i(F(\theta_{-i},\theta_i'),\theta_i')\mid\theta_i') \text{ for some } \theta_i'\in\Theta_i
    \end{cases}.\]

    \underline{Rule 3:} Suppose neither rule 1 nor rule 2 applies. For $1\leq k\leq 4,$ let $m_{i}^k$ denote the $k$th entry of the message of agent $i.$ Then $h(m)=m_{i^*}^3(\theta^m)$, where $i^*$ is determined by the following. Let $I_0$ be the set of agents $i\in I$ that choose to specify an integer vector in their message, so $I_0=\{i\mid m_i^2\neq \varnothing\}$. For $i\in I_0$ denote $m_i^2$ by $v^i,$ and let the $l$th coordinate of the vector $v^i$ be $v_l^j.$ Let \[J(i)=\#\{j\in I_0\mid v_\ell^i=v_i^j \text{ for an integer } \ell \text{ such that } n+(j-1)S<\ell\leq n+jS\}.\] If there is $i$ such that $J(i)>J(k)$ for all $k\in I_0$ and  $k\neq i$, then $i^*=i.$ Otherwise, $i^*=\min_i(I_0).$

Rules 2 and 3 ensure that in an equilibrium $s$ of $\Gamma$, Rule 1 must apply, as otherwise one agent can unilaterally deviate to get her most preferred outcome. If rule 1 applies and $s$ is an equilibrium of $\Gamma,$ it must be the case that $(s_i^*)_{i\in I}$ — where $s_i^*(\theta_i)$ is the first coordinate of $s_i(\theta_i)$ for all $i$ and $\theta_i\in \Theta_i$ — is an equilibrium of the direct mechanism. Rule 2 ensures that $s=(s_i)_{i\in I}$ where $s_i(\theta_i)$ is of the form $(\theta_i,\varnothing, \cdot, \varnothing)$ for all $i$ and $\theta_i\in \Theta_i$ is a BNE of $\Gamma$, so $\Gamma$ clearly partially implements $F.$ If $\mathcal{A}_F=\mathcal{A}_F^u$ then rule 2 ensures that no $s$ such that $(s_i^*)_{i\in I}$ is an undesirable equilibrium of the direct mechanism constitutes an equilibrium of $\Gamma.$ Thus $\Gamma$ fully implements $F$ whenever $\mathcal{A}_F=\mathcal{A}_F^u$. 
\end{proof}

\textbf{Proof of Lemma 2.}
\begin{proof}
Suppose agent $i$ can potentially undermine $\alpha$. Therefore truth-telling is a dominant strategy for agent $i$ and there exist a type $\theta_i\in\Theta_i$ as well as distinct type profiles $\theta_{-i}^1,\theta_{-i}^2, \theta_{-i}^3, \theta_{-i}^4\in\Theta_{-i}$ such that $\alpha_{-i}(\theta_{-i}^4)\neq \theta_{-i}^1,\alpha_{-i}(\theta_{-i}^2)\neq \theta_{-i}^3$, and $\alpha_{-i}(\theta_{-i}^j)=\theta_{-i}^{j+1}$ for $j=1,3$. Furthermore, $$u_i(F(\theta_{-i}^j,\theta_i'),\theta_i')\geq u_i(F(\theta_{-i}^{j+1},\theta_i'),\theta_i')$$ for $j=1,3$ and all $\theta_i'\in\Theta_i,$ and each inequality is strict for some type $\theta_i.$ Let $\sigma_i\in\{0,1\}$ be a binary signal such that if $\theta_{-i}\in\{\theta_{-i}^1,\theta_{-i}^4\}\cup \alpha_{-i}^{-1}(\theta_{-i}^3)$, then $P(\sigma_i=0\mid \theta_{-i})=\tau$ for some accuracy $0.5<\tau<1$. Additionally, if $\theta_{-i}\notin\{\theta_{-i}^1,\theta_{-i}^4\}\cup \alpha_{-i}^{-1}(\theta_{-i}^3)$, then $P(\sigma_i=1\mid \theta_{-i})=\tau.$ With slight abuse of notation we write $\{\theta_{-i}^1,\theta_{-i}^4, \alpha_{-i}^{-1}(\theta_{-i}^3)\}$ to mean $\{\theta_{-i}^1,\theta_{-i}^4\}\cup \alpha_{-i}^{-1}(\theta_{-i}^3).$ Since agent $i$ can potentially undermine $\alpha$, we have from the definition that $\{\alpha_{-i}^{-1}(\theta_{-i}^1),\theta_{-i}^2,\theta_{-i}^3\}\subset \Theta_{-i}\setminus \{\theta_{-i}^1,\theta_{-i}^4,\alpha_{-i}^{-1}(\theta_{-i}^3)\}$. Thus, upon receiving $\sigma_i=0$, agent $i$ will place probability \[\tau\mu(\theta_{-i}'\mid\theta_i',\theta_{-i}\in\{\theta_{-i}^1,\theta_{-i}^4, \alpha_{-i}^{-1}(\theta_{-i}^3)\})\] on all $\theta_{-i}'\in \{\theta_{-i}^1,\theta_{-i}^4, \alpha_{-i}^{-1}(\theta_{-i}^3)\}$ and probability \[(1-\tau)\mu(\theta_{-i}'\mid\theta_i',\theta_{-i}\notin\{\theta_{-i}^1,\theta_{-i}^4, \alpha_{-i}^{-1}(\theta_{-i}^3)\})\] on all $\theta_{-i}'\notin \{\theta_{-i}^1,\theta_{-i}^4, \alpha_{-i}^{-1}(\theta_{-i}^3)\}.$ Here $\mu(\cdot\mid\theta_i',\theta_{-i}\in\{\theta_{-i}^1,\theta_{-i}^4, \alpha_{-i}^{-1}(\theta_{-i}^3)\})$ represents the beliefs agent $i$ would have over $\Theta_{-i}$ conditional on knowing her type is $\theta_i'$, and the type profile $\theta_{-i}$ of other agents is in the set $\{\theta_{-i}^1,\theta_{-i}^4, \alpha_{-i}^{-1}(\theta_{-i}^3)\}.$ We will show that there exists an accurate enough $\tau$ such that agent $i$ can undermine $\alpha$.

To prove the lemma, $\alpha$ must be undermined by agent $i$ in either realization of the signal, so first suppose $\sigma_i=0.$ Now let $y:\Theta_{-i}\to A$ be defined by $y(\theta_{-i}^1)=F(\theta_{-i}^2,\alpha_i(\theta_i)), y(\theta_{-i}^2)=F(\theta_{-i}^1,\alpha_i(\theta_i))$, and $y(\theta_{-i})=F(\theta_{-i},\alpha_i(\theta_i))$ for all $\theta_{-i}\neq \theta_{-i}^1,\theta_{-i}^2.$ This construction of $y$ is a general formulation of the alternate SCFs constructed for the main example. It first needs to be shown that \[\mathbb E(u_i(F(\theta_{-i},\theta_i'),\theta_i')\mid\theta_i',\sigma_i=0)\geq \mathbb E(u_i(y(\theta_{-i}),\theta_i')\mid\theta_i',\sigma_i=0)\] for all $\theta_i'\in\Theta_i.$ To see this, observe that the following holds for sufficiently accurate $\tau$.
\begin{align*}
    \mathbb E(u_i(F&(\theta_{-i},\theta_i'),\theta_i')\mid\theta_i',\sigma_i=0)= \tau \mathbb E(u_i(F(\theta_{-i},\theta_i'),\theta_i')\mid\theta_i',\theta_{-i}\in\{\theta_{-i}^1,\theta_{-i}^4,\alpha_{-i}^{-1}(\theta_{-i}^3)\})\\ &+ (1-\tau) \mathbb E(u_i(F(\theta_{-i},\theta_i'),\theta_i')\mid\theta_i',\theta_{-i}\notin \{\theta_{-i}^1,\theta_{-i}^4,\alpha_{-i}^{-1}(\theta_{-i}^3)\})\\
    &\geq \tau \mu(\theta_{-i}^1\mid \theta_i',\theta_{-i}\in \{\theta_{-i}^1,\theta_{-i}^4,\alpha_{-i}^{-1}(\theta_{-i}^3)\})u_i(F(\theta_{-i}^2,\theta_i'),\theta_i')\\
    &+ \tau \mu(\theta_{-i}^4\mid \theta_i',\theta_{-i}\in \{\theta_{-i}^1,\theta_{-i}^4,\alpha_{-i}^{-1}(\theta_{-i}^3)\})u_i(F(\theta_{-i}^4,\theta_i'),\theta_i')\\
    &+ \tau \sum_{\theta_{-i}'\in\alpha_{-i}^{-1}(\theta_{-i}^3)} \mu(\theta_{-i}'\mid \theta_i',\theta_{-i}\in \{\theta_{-i}^1,\theta_{-i}^4,\alpha_{-i}^{-1}(\theta_{-i}^3)\})u_i(F(\theta_{-i}',\theta_i'),\theta_i')\\
    &+(1-\tau)\mu(\theta_{-i}^2\mid \theta_i',\theta_{-i}\notin \{\theta_{-i}^1,\theta_{-i}^4,\alpha_{-i}^{-1}(\theta_{-i}^3)\})u_i(F(\theta_{-i}^1,\theta_i'),\theta_i')\\
    &+(1-\tau)\sum_{\substack{\theta_{-i}'\in \Theta_{-i}\setminus \{\theta_{-i}^1,\theta_{-i}^4,\alpha_{-i}^{-1}(\theta_{-i}^3)\}\\ \theta_{-i}'\neq \theta_{-i}^2}}\mu(\theta_{-i}'\mid \theta_i',\theta_{-i}\notin\{\theta_{-i}^1,\theta_{-i}^4,\alpha_{-i}^{-1}(\theta_{-i}^3)\})u_i(F(\theta_{-i}',\theta_i'),\theta_i')\\
    &\geq \tau \mu(\theta_{-i}^1\mid \theta_i',\theta_{-i}\in \{\theta_{-i}^1,\theta_{-i}^4,\alpha_{-i}^{-1}(\theta_{-i}^3)\})u_i(F(\theta_{-i}^2,\alpha_i(\theta_i)),\theta_i')\\
    &+ \tau \mu(\theta_{-i}^4\mid \theta_i',\theta_{-i}\in \{\theta_{-i}^1,\theta_{-i}^4,\alpha_{-i}^{-1}(\theta_{-i}^3)\})u_i(F(\theta_{-i}^4,\alpha_i(\theta_i)),\theta_i')\\
    &+ \tau \sum_{\theta_{-i}'\in\alpha_{-i}^{-1}(\theta_{-i}^3)} \mu(\theta_{-i}'\mid \theta_i',\theta_{-i}\in \{\theta_{-i}^1,\theta_{-i}^4,\alpha_{-i}^{-1}(\theta_{-i}^3)\})u_i(F(\theta_{-i}',\alpha_i(\theta_i)),\theta_i')\\
    &+(1-\tau)\mu(\theta_{-i}^2\mid \theta_i',\theta_{-i}\notin \{\theta_{-i}^1,\theta_{-i}^4,\alpha_{-i}^{-1}(\theta_{-i}^3)\})u_i(F(\theta_{-i}^1,\alpha_i(\theta_i)),\theta_i')\\
    &+(1-\tau)\sum_{\substack{\theta_{-i}'\in \Theta_{-i}\setminus \{\theta_{-i}^1,\theta_{-i}^4,\alpha_{-i}^{-1}(\theta_{-i}^3)\}\\ \theta_{-i}'\neq \theta_{-i}^2}}\mu(\theta_{-i}'\mid \theta_i',\theta_{-i}\notin\{\theta_{-i}^1,\theta_{-i}^4,\alpha_{-i}^{-1}(\theta_{-i}^3)\})u_i(F(\theta_{-i}',\alpha_i(\theta_i)),\theta_i')\\
    &= \mathbb E(u_i(y(\theta_{-i}),\theta_i')\mid\theta_i',\sigma_i=0),
\end{align*}
as desired. The first inequality is implied by our assumption that agent $i$ always prefers $\theta_{-i}^1$ to $\theta_{-i}^2$ and holds whenever \[\tau \mu(\theta_{-i}^1\mid \theta_i',\theta_{-i}\in \{\theta_{-i}^1,\theta_{-i}^4,\alpha_{-i}^{-1}(\theta_{-i}^3)\})\geq (1-\tau)\mu(\theta_{-i}^2\mid \theta_i',\theta_{-i}\notin \{\theta_{-i}^1,\theta_{-i}^4,\alpha_{-i}^{-1}(\theta_{-i}^3)\}),\] which can be done for sufficiently large $\tau.$ The second inequality is because truth-telling is a dominant strategy for agent $i.$ The last equality follows from the definition of $y.$ Thus, the desired inequality holds.

It also needs to be shown that there is sufficiently large $\tau$ such that \[\mathbb E(u_i(F(\alpha(\theta_{-i},\theta_i)),\theta_i)\mid\theta_i,\sigma_i=0)<\mathbb E(u_i(y(\alpha_{-i}(\theta_{-i})),\theta_i)\mid\theta_i,\sigma_i=0).\] To see this, as before let $\alpha_{-i}^{-1}(\theta_{-i}^1)$ be the (possibly empty) set of $\theta_{-i}$ such that $\alpha_{-i}(\theta_{-i})=\theta_{-i}^1$. Note that our assumptions imply that $\alpha_{-i}^{-1}(\theta_{-i}^1)$ is disjoint from $\{\theta_{-i}^1,\theta_{-i}^4,\alpha_{-i}^{-1}(\theta_{-i}^3)\}$. So,
\begin{align*}
\mathbb E(u_i&(F(\alpha (\theta_{-i},\theta_i)),\theta_i)\mid\theta_i,\sigma_i=0) = \tau\mu(\theta_{-i}^1\mid \theta_i,\theta_{-i}\in \{\theta_{-i}^1,\theta_{-i}^4,\alpha_{-i}^{-1}(\theta_{-i}^3)\}) u_i(F(\theta_{-i}^2,\alpha_i(\theta_i)),\theta_i)\\
    &+ \tau \mu(\theta_{-i}^4\mid \theta_i,\theta_{-i}\in \{\theta_{-i}^1,\theta_{-i}^4,\alpha_{-i}^{-1}(\theta_{-i}^3)\})u_i(F(\alpha_{-i}(\theta_{-i}^4),\alpha_i(\theta_i)),\theta_i)\\
    &+ \tau \sum_{\theta_{-i}'\in \alpha_{-i}^{-1}(\theta_{-i}^3)}\mu(\theta_{-i}'\mid \theta_i,\theta_{-i}\in \{\theta_{-i}^1,\theta_{-i}^4,\alpha_{-i}^{-1}(\theta_{-i}^3)\}) u_i(F(\theta_{-i}^3,\alpha_i(\theta_i)),\theta_i) \\
    &+ (1-\tau) \sum_{\theta_{-i}'\in \alpha_{-i}^{-1}(\theta_{-i}^1)}\mu(\theta_{-i}'\mid \theta_i,\theta_{-i}\notin \{\theta_{-i}^1,\theta_{-i}^4,\alpha_{-i}^{-1}(\theta_{-i}^3)\}) u_i(F(\theta_{-i}^1,\alpha_i(\theta_i)),\theta_i) \\
    &+(1-\tau) \sum_{\substack{\theta_{-i}'\notin\alpha_{-i}^{-1}(\theta_{-i}^1)\\ \theta_{-i}'\notin \{\theta_{-i}^1,\theta_{-i}^4,\alpha_{-i}^{-1}(\theta_{-i}^3)\}}}\mu(\theta_{-i}'\mid \theta_i,\theta_{-i}\notin \{\theta_{-i}^1,\theta_{-i}^4,\alpha_{-i}^{-1}(\theta_{-i}^3)\}) u_i(F(\alpha_{-i}(\theta_{-i}'),\alpha_i(\theta_i)),\theta_i).
\end{align*}

Additionally, \begin{align*}
\mathbb E(u_i&(y(\alpha_{-i}(\theta_{-i})),\alpha_i(\theta_i))\mid\theta_i,\sigma_i=0) = \tau\mu(\theta_{-i}^1\mid \theta_i,\theta_{-i}\in \{\theta_{-i}^1,\theta_{-i}^4,\alpha_{-i}^{-1}(\theta_{-i}^3)\}) u_i(F(\theta_{-i}^1,\alpha_i(\theta_i)),\theta_i)\\
    &+ \tau \mu(\theta_{-i}^4\mid \theta_i,\theta_{-i}\in \{\theta_{-i}^1,\theta_{-i}^4,\alpha_{-i}^{-1}(\theta_{-i}^3)\})u_i(F(\alpha_{-i}(\theta_{-i}^4),\alpha_i(\theta_i)),\theta_i)\\
    &+ \tau \sum_{\theta_{-i}'\in \alpha_{-i}^{-1}(\theta_{-i}^3)}\mu(\theta_{-i}'\mid \theta_i,\theta_{-i}\in \{\theta_{-i}^1,\theta_{-i}^4,\alpha_{-i}^{-1}(\theta_{-i}^3)\}) u_i(F(\theta_{-i}^3,\alpha_i(\theta_i)),\theta_i) \\
    &+ (1-\tau) \sum_{\theta_{-i}'\in \alpha_{-i}^{-1}(\theta_{-i}^1)}\mu(\theta_{-i}'\mid \theta_i,\theta_{-i}\notin \{\theta_{-i}^1,\theta_{-i}^4,\alpha_{-i}^{-1}(\theta_{-i}^3)\}) u_i(F(\theta_{-i}^2,\alpha_i(\theta_i)),\theta_i) \\
    &+(1-\tau) \sum_{\substack{\theta_{-i}'\notin\alpha_{-i}^{-1}(\theta_{-i}^1)\\ \theta_{-i}'\notin \{\theta_{-i}^1,\theta_{-i}^4,\alpha_{-i}^{-1}(\theta_{-i}^3)\}}}\mu(\theta_{-i}'\mid \theta_i,\theta_{-i}\notin \{\theta_{-i}^1,\theta_{-i}^4,\alpha_{-i}^{-1}(\theta_{-i}^3)\}) u_i(y(\alpha_{-i}(\theta_{-i}')),\theta_i).
\end{align*}

Comparing these expressions, first observe that \[\sum_{\substack{\theta_{-i}'\notin\alpha_{-i}^{-1}(\theta_{-i}^1)\\ \theta_{-i}'\notin \{\theta_{-i}^1,\theta_{-i}^4,\alpha_{-i}^{-1}(\theta_{-i}^3)\}}}\mu(\theta_{-i}'\mid \theta_i,\theta_{-i}\notin \{\theta_{-i}^1,\theta_{-i}^4,\alpha_{-i}^{-1}(\theta_{-i}^3)\}) u_i(F(\alpha_{-i}(\theta_{-i}'),\alpha_i(\theta_i)),\theta_i)\] \[\leq \sum_{\substack{\theta_{-i}'\notin\alpha_{-i}^{-1}(\theta_{-i}^1)\\ \theta_{-i}'\notin \{\theta_{-i}^1,\theta_{-i}^4,\alpha_{-i}^{-1}(\theta_{-i}^3)\}}}\mu(\theta_{-i}'\mid \theta_i,\theta_{-i}\notin \{\theta_{-i}^1,\theta_{-i}^4,\alpha_{-i}^{-1}(\theta_{-i}^3)\}) u_i(y(\alpha_{-i}(\theta_{-i}')),\theta_i).\] This is because the only possible $\theta_{-i}$ for which terms in the sum will differ are for $\theta_{-i}\in \alpha_{-i}^{-1}(\theta_{-i}^2)$ such that $\theta_{-i}\neq \theta_{-i}^1.$ In the former summation, the corresponding term will be \[K_2=\mu(\theta_{-i}'\mid \theta_i,\theta_{-i}\notin \{\theta_{-i}^1,\theta_{-i}^4,\alpha_{-i}^{-1}(\theta_{-i}^3)\}) u_i(F(\theta_{-i}^2,\alpha_i(\theta_i)),\theta_i).\] The corresponding term in the latter summation, which is  \[K_1=\mu(\theta_{-i}'\mid \theta_i,\theta_{-i}\notin \{\theta_{-i}^1,\theta_{-i}^4,\alpha_{-i}^{-1}(\theta_{-i}^3)\}) u_i(F(\theta_{-i}^1,\alpha_i(\theta_i)),\theta_i).\] If $\alpha_i(\theta_i)=\theta_i,$ then $K_2<K_1$ is automatic from the requirement that agent $i$ can potentially undermine $\alpha$. If $\alpha_i(\theta_i)\neq\theta_i$, then $u_i$ must satisfy \[u_i(F(\theta_{-i}^2,\alpha_i(\theta_i)), \theta_i)=u_i(F(\theta_{-i}^2,\theta_i),\theta_i)\] and \[u_i(F(\theta_{-i}^1,\alpha_i(\theta_i)),\theta_i)=u_i(F(\theta_{-i}^1,\theta_i),\theta_i)\] because truth-telling is a dominant strategy for agent $i.$ Thus the assumption that agent $i$ can potentially undermine $\alpha$ also implies $K_2<K_1.$

Thus, to guarantee \[\mathbb E(u_i(F(\alpha(\theta_{-i},\theta_i)),\theta_i)\mid\theta_i,\sigma_i=0)<\mathbb E(u_i(y(\alpha_{-i}(\theta_{-i})),\theta_i)\mid\theta_i,\sigma_i=0),\] it  suffices to pick $\tau$ sufficiently large so that \[\tau\mu(\theta_{-i}^1\mid \theta_i,\theta_{-i}\in \{\theta_{-i}^1,\theta_{-i}^4,\alpha_{-i}^{-1}(\theta_{-i}^3)\}) u_i(F(\theta_{-i}^2,\alpha_i(\theta_i)),\theta_i)\]\[+(1-\tau) \sum_{\theta_{-i}'\in \alpha_{-i}^{-1}(\theta_{-i}^1)}\mu(\theta_{-i}'\mid \theta_i,\theta_{-i}\notin \{\theta_{-i}^1,\theta_{-i}^4,\alpha_{-i}^{-1}(\theta_{-i}^3)\}) u_i(F(\theta_{-i}^1,\alpha_i(\theta_i)),\theta_i)\]\[< \tau\mu(\theta_{-i}^1\mid \theta_i,\theta_{-i}\in \{\theta_{-i}^1,\theta_{-i}^4,\alpha_{-i}^{-1}(\theta_{-i}^3)\}) u_i(F(\theta_{-i}^1,\alpha_i(\theta_i)),\theta_i)\]\[+(1-\tau) \sum_{\theta_{-i}'\in \alpha_{-i}^{-1}(\theta_{-i}^1)}\mu(\theta_{-i}'\mid \theta_i,\theta_{-i}\notin \{\theta_{-i}^1,\theta_{-i}^4,\alpha_{-i}^{-1}(\theta_{-i}^3)\}) u_i(F(\theta_{-i}^2,\alpha_i(\theta_i)),\theta_i),\] which is again possible by the condition that agent $i$ can potentially undermine $\alpha$ and similar reasoning as above. Thus $\alpha$ is undermined in the signal realization $\sigma_i=0.$

The proof is analogous for the realization $\sigma_i=1.$
\end{proof}
\textbf{Proof of Theorem 1.}
\begin{proof}
Let $\mathcal{A}_F\setminus \mathcal{A}_F^u=\{\alpha_1,\dots,\alpha_k\}$ and $j_\ell$ be the agent that can potentially undermine $\alpha_\ell$ for all $1\leq \ell\leq k$. By lemma \ref{lem:main}, for each $\ell$ there is a binary signal $\sigma_{j_\ell}$ such that agent $j_\ell$ can undermine $\alpha_\ell$ upon receiving $\sigma_{j_\ell}$. By assumption each $j_\ell$ is distinct and $\{j_1,\dots,j_k\}\subset J.$ For all $i\in I\setminus \{j_1,\dots,j_k\}$ let $\sigma_i$ be uninformative. Lemma \ref{lem:revelation} directly implies that $F$ is fully implementable with respect to these signals $\{\sigma_i\}_{i\in I}$ as long as the set of equilibria in the direct mechanism is unchanged after the agents $j_1,\dots,j_k$ receive their signals. More precisely, if we take $\mathcal{B}_F$ to be the set of Bayes Nash equilibria of the direct mechanism when agents only know the common prior $\mu\in\Delta(\Theta)$, we want to show that $s\in \mathcal{B}_F$ if and only if $s$ is an equilibrium of the direct mechanism after the agents $j_1,\dots,j_k$ receive their signals.

To see this, let $s=(s_1,\dots,s_n)$ be a strategy profile. Then $s$ is an equilibrium in the direct mechanism when no agent has additional information about other agents iff for all $i$ and $\theta_i,\theta_i'\in\Theta_{i}$ we have \[\sum_{\theta_{-i}}u(F(s_i(\theta_i),s_{-i}(\theta_{-i})),\theta_i)\mu(\theta_{-i}\mid\theta_i)\geq \sum_{\theta_{-i}}u(F(\theta_i',s_{-i}(\theta_{-i})),\theta_i)\mu(\theta_{-i}\mid\theta_i).\tag{1}\] On the other hand, $s$ is an equilibrium in the direct mechanism after agents receive the signal realizations $\{\sigma_i\}_{i\in I}$ iff for all $i\in I$,  $\theta_i,\theta_i'\in\Theta_{i}$, and signal realizations $x$ in the range of $\sigma_i$, we have \[\sum_{\theta_{-i}}u(F(s_i(\theta_i),s_{-i}(\theta_{-i})),\theta_i)\mu(\theta_{-i}\mid\theta_i,\sigma_i=x)\geq \sum_{\theta_{-i}}u(F(\theta_i',s_{-i}(\theta_{-i})),\theta_i)\mu(\theta_{-i}\mid\theta_i, \sigma_i=x).\tag{2}\] When $i\notin \{j_1,\dots,j_k\}$ the signal $\sigma_i$ is uninformative so the inequalities don't change between (1) and (2). If $i\in\{j_1,\dots,j_k\}$, then truth-telling is a dominant strategy for agent $i$, so it must be that \[u(F(\theta_i,s_{-i}(\theta_{-i})),\theta_i)\geq u(F(\theta_i',s_{-i}(\theta_{-i})),\theta_i)\] for all $\theta_i,\theta_i'\in \Theta_i$, since $\mu$ and $\sigma_i$ have full support. Thus both (1) and (2) imply \[u(F(s_i(\theta_i),s_{-i}(\theta_{-i})),\theta_i)=u(F(\theta_i,s_{-i}(\theta_{-i})),\theta_i)\] for all $i\in \{j_1,\dots,j_k\}$ and  $\theta_i\in \Theta_i$. This completes the proof.
\end{proof}
\clearpage
\begin{singlespace}
\bibliographystyle{aer}
\bibliography{sources.bib}

@article{jackson91,
	title = {Bayesian {Implementation}},
	volume = {59},
	issn = {0012-9682},
	url = {https://www.jstor.org/stable/2938265},
	doi = {10.2307/2938265},
	number = {2},
	urldate = {2023-09-23},
	journal = {Econometrica},
	author = {Jackson, Mathew O.},
	year = {1991},
	note = {Publisher: [Wiley, Econometric Society]},
	pages = {461--477},
}

@article{BM22,
	title = {Optimal {Information} {Disclosure} in {Classic} {Auctions}},
	volume = {4},
	url = {https://www.aeaweb.org/articles?id=10.1257/aeri.20210504},
	doi = {10.1257/aeri.20210504},
	language = {en},
	number = {3},
	urldate = {2023-09-23},
	journal = {American Economic Review: Insights},
	author = {Bergemann, Dirk and Heumann, Tibor and Morris, Stephen and Sorokin, Constantine and Winter, Eyal},
	month = sep,
	year = {2022},
	pages = {371--388},
}

@article{bergemann_robust_2005,
	title = {Robust {Mechanism} {Design}},
	volume = {73},
	issn = {0012-9682},
	url = {https://www.jstor.org/stable/3598751},
	number = {6},
	urldate = {2023-09-23},
	journal = {Econometrica},
	author = {Bergemann, Dirk and Morris, Stephen},
	year = {2005},
	note = {Publisher: [Wiley, Econometric Society]},
	pages = {1771--1813},
}

@article{Maskin99,
	title = {Nash {Equilibrium} and {Welfare} {Optimality}},
	volume = {66},
	issn = {0034-6527},
	url = {https://www.jstor.org/stable/2566947},
	number = {1},
	urldate = {2024-01-30},
	journal = {The Review of Economic Studies},
	author = {Maskin, Eric},
	year = {1999},
	note = {Publisher: [Oxford University Press, Review of Economic Studies, Ltd.]},
	pages = {23--38},
}

@article{PS89,
	title = {Implementation with {Incomplete} {Information} in {Exchange} {Economies}},
	volume = {57},
	issn = {0012-9682},
	url = {https://www.jstor.org/stable/1912575},
	doi = {10.2307/1912575},
	number = {1},
	urldate = {2024-01-30},
	journal = {Econometrica},
	author = {Palfrey, Thomas R. and Srivastava, Sanjay},
	year = {1989},
	note = {Publisher: [Wiley, Econometric Society]},
	pages = {115--134},
}

@article{PS86,
	title = {Implementation in differential information economies},
	volume = {39},
	issn = {0022-0531},
	url = {https://www.sciencedirect.com/science/article/pii/0022053186900189},
	doi = {10.1016/0022-0531(86)90018-9},
	number = {1},
	urldate = {2024-01-30},
	journal = {Journal of Economic Theory},
	author = {Postlewaite, Andrew and Schmeidler, David},
	month = jun,
	year = {1986},
	pages = {14--33},
}

@article{GP23,
	title = {Informationally {Simple} {Incentives}},
	volume = {131},
	issn = {0022-3808},
	url = {https://www.journals.uchicago.edu/doi/abs/10.1086/722089},
	doi = {10.1086/722089},
	number = {3},
	urldate = {2024-01-31},
	journal = {Journal of Political Economy},
	author = {Gleyze, Simon and Pernoud, Agathe},
	month = mar,
	year = {2023},
	note = {Publisher: The University of Chicago Press},
	pages = {802--837},
}

@article{BV02,
	title = {Information {Acquisition} and {Efficient} {Mechanism} {Design}},
	volume = {70},
	issn = {0012-9682},
	url = {https://www.jstor.org/stable/2692306},
	number = {3},
	urldate = {2024-01-31},
	journal = {Econometrica},
	author = {Bergemann, Dirk and Välimäki, Juuso},
	year = {2002},
	note = {Publisher: [Wiley, Econometric Society]},
	pages = {1007--1033},
}

@article{ganuza04,
	title = {Ignorance {Promotes} {Competition}: {An} {Auction} {Model} with {Endogenous} {Private} {Valuations}},
	volume = {35},
	issn = {0741-6261},
	shorttitle = {Ignorance {Promotes} {Competition}},
	url = {https://www.jstor.org/stable/1593709},
	doi = {10.2307/1593709},
	number = {3},
	urldate = {2024-02-24},
	journal = {The RAND Journal of Economics},
	author = {Ganuza, Juan-José},
	year = {2004},
	note = {Publisher: [RAND Corporation, Wiley]},
	pages = {583--598},
}

@article{BP07,
	title = {Information structures in optimal auctions},
	volume = {137},
	issn = {0022-0531},
	url = {https://www.sciencedirect.com/science/article/pii/S0022053107000324},
	doi = {10.1016/j.jet.2007.02.001},
	number = {1},
	urldate = {2024-02-24},
	journal = {Journal of Economic Theory},
	author = {Bergemann, Dirk and Pesendorfer, Martin},
	month = nov,
	year = {2007},
	pages = {580--609},
}

@article{ES07,
	title = {Optimal {Information} {Disclosure} in {Auctions} and the {Handicap} {Auction}},
	volume = {74},
	issn = {0034-6527},
	url = {https://www.jstor.org/stable/4626158},
	number = {3},
	urldate = {2024-02-24},
	journal = {The Review of Economic Studies},
	author = {Eső, Péter and Szentes, Balázs},
	year = {2007},
	note = {Publisher: [Oxford University Press, Review of Economic Studies, Ltd.]},
	pages = {705--731},
}

@article{LS17,
	title = {Discriminatory {Information} {Disclosure}},
	volume = {107},
	issn = {0002-8282},
	url = {https://www.aeaweb.org/articles?id=10.1257/aer.20151743},
	doi = {10.1257/aer.20151743},
	language = {en},
	number = {11},
	urldate = {2024-02-24},
	journal = {American Economic Review},
	author = {Li, Hao and Shi, Xianwen},
	month = nov,
	year = {2017},
	pages = {3363--3385}
}

@article{SV2010,
	series = {Mathematical {Economics}: {Special} {Issue} in honour of {Andreu} {Mas}-{Colell}, {Part} 1},
	title = {Multiplicity of mixed equilibria in mechanisms: {A} unified approach to exact and approximate implementation},
	volume = {46},
	issn = {0304-4068},
	shorttitle = {Multiplicity of mixed equilibria in mechanisms},
	url = {https://www.sciencedirect.com/science/article/pii/S0304406810000546},
	doi = {10.1016/j.jmateco.2010.06.002},
	number = {5},
	urldate = {2025-06-12},
	journal = {Journal of Mathematical Economics},
	author = {Serrano, Roberto and Vohra, Rajiv},
	month = sep,
	year = {2010},
	pages = {775--785},
}

@article{KamenicaGentzkow11,
author="Emir Kamenica and Matthew Gentzkow",
title="Bayesian Persuasion",
journal="American Economic Review",
volume="101",
number="6",
pages="2590-2615",
year="2011",
DOI="10.1257/aer.101.6.2590",
}

@article{rayo2010optimal,
  title={Optimal Information Disclosure},
  author={Rayo, Luis and Segal, Ilya},
  journal={Journal of Political Economy},
  volume={118},
  number={5},
  pages={949--987},
  year={2010},
  doi={10.1086/657922},
  url={https://www.jstor.org/stable/10.1086/657922}
}

@article{roesler_buyer-optimal_2017,
	title = {Buyer-{Optimal} {Learning} and {Monopoly} {Pricing}},
	volume = {107},
	issn = {0002-8282},
	url = {https://www.aeaweb.org/articles?id=10.1257/aer.20160145},
	doi = {10.1257/aer.20160145},
	language = {en},
	number = {7},
	urldate = {2025-03-08},
	journal = {American Economic Review},
	author = {Roesler, Anne-Katrin and Szentes, Balázs},
	month = jul,
	year = {2017},
	pages = {2072--2080},
}
\end{singlespace}

\end{document}